\documentclass[12pt,letterpaper]{article}

\usepackage[english]{babel}
\usepackage[T1]{fontenc}

\usepackage{fontspec} 
\usepackage{amsmath, amssymb, mathtools} 
\usepackage{bm} 
\DeclareMathOperator*{\argmax}{arg\,max}

\usepackage[margin=1in]{geometry}

\usepackage[usenames,dvipsnames]{xcolor}
\usepackage[bookmarks=true,colorlinks=true,linkcolor=MyBlue,citecolor=MyRed,urlcolor=MyRed]{hyperref}

\usepackage{graphicx}
\usepackage{caption}
\usepackage{subcaption}
\usepackage{enumerate}
\usepackage{enumitem}  

\usepackage{tikz}
\usepackage{pgfplots}
\pgfplotsset{compat=1.18} 
\usetikzlibrary{intersections}

\usepackage{comment}

\usepackage[disable]{todonotes}

\usepackage{amsthm}  
\newtheorem{theorem}{Theorem}
\newtheorem{definition}{Definition}
\newtheorem{example}{Example}
\newtheorem{corollary}{Corollary}
\newtheorem{remark}{Remark}

\newtheorem{proposition}{Proposition}

\usepackage[authoryear]{natbib} 
\usepackage{epsfig} 
\usepackage{url} 
\usepackage{float} 
\usepackage{cleveref} 

\usepackage{array} 
\usepackage{soul} 

\definecolor{MyBlue}{rgb}{0,0,0.65}
\definecolor{MyRed}{rgb}{0.6,0,0.1}

\usepackage{sectsty}
\sectionfont{\color{MyBlue}}  
\subsectionfont{\color{MyBlue}}

\usepackage{setspace}
\relax
\makeatother

\begin{document}

\title{Randomization and Efficiency under Incomplete Information\thanks{We thank Sarah Auster and Tristan Tomala for valuable feedback and guidance, and the audiences at HEC Paris, University of Bonn and  AMES 2026 for helpful comments. Itai Arieli acknowledges financial support from the Israel Science Foundation, Grant No. 2039/24. Yakov Babichenko acknowledges financial support from the Israel Science Foundation, Grant No. 2072072. Atulya Jain  acknowledges funding from the DATAIA Convergence Institute (ANR-17-CONV-0003), a joint initiative of HEC Paris and Hi! PARIS, from the Programme d’Investissement d’Avenir (ANR-18-EURE-0005 / EUR DATA EFM), and from the German Research Foundation (DFG) under Germany’s Excellence Strategy – EXC 2126/2–390838866.   Rann Smorodinsky acknowledges funding from the German Research Foundation (DFG), research grant 514505843 and from the Israeli Science Foundation (ISF), research grant 2646/22.  }

\author{
Itai Arieli\footnote{ Technion:  Israel Institute of Technology and University of Toronto. Email: \texttt{iarieli@technion.ac.il}.} \and
Yakov Babichenko\footnote{Technion: Israel Institute of Technology. Email: \texttt{yakovbab@technion.ac.il}.} \and
Atulya Jain\footnote{University of Bonn. Email: \texttt{ajain@uni-bonn.de}.} \and
Rann Smorodinsky\footnote{Technion: Israel Institute of Technology. Email: \texttt{rann@technion.ac.il}.}
}}

\date{\today}

\maketitle

\vspace{-10mm}
\begin{abstract}

We study ex-ante Pareto efficiency of feasible outcomes in games with incomplete information. Our main observation is that excessive randomization leads to inefficiency. With a nonatomic prior, generically, efficient outcomes are pure almost surely. In finite state spaces, the restriction is weaker: generically, the total number of actions used across states must be strictly less than the sum of the number of players and states. We apply these results to three environments. In an allocation problem without transfers, the welfare-maximizing incentive-compatible outcome is inefficient. In cheap talk, generically, an outcome is efficient only if it is pure; with state-independent sender payoffs, it is efficient if and only if the sender’s most preferred action is induced with certainty. In Bayesian persuasion with one safe action and several risky actions, outcomes are generically inefficient across broad sets of priors and receiver preferences.
\end{abstract}




\section{Introduction}

A central question in economics is how individual incentives interact with social welfare. A natural benchmark for studying this relationship is Pareto efficiency---outcomes in which no individual can be made strictly better off without making another worse off. However, in many real-world settings, individuals make decisions under incomplete information, making it unclear when efficiency can be achieved. They face uncertainty about the underlying state of the world and about what others know. Information could, in principle, improve welfare, yet individual incentives often stand in the way of efficiency.

The question is especially relevant to information design and mechanism design. In information design, uncertainty concerns an underlying state of the world, and the designer chooses how information about that state is revealed. In mechanism design, uncertainty concerns agents’ privately known types, and the designer chooses rules that map their reports into allocations. In both cases, the resulting outcome maps the underlying uncertainty\textemdash states or type profiles\textemdash into distributions over actions or allocations. Understanding when such outcomes are Pareto efficient is therefore key to understanding the scope and limits of design under incomplete information.  We evaluate outcomes by players’ ex-ante expected payoffs and assess Pareto efficiency relative to the set of all feasible outcomes.

We start with finite state spaces. We provide a simple necessary condition for Pareto efficiency that depends only on the number of actions used across states. Generically, an outcome is efficient only if the total number of actions used across all states is strictly less than the sum of the number of players and the number of states (\Cref{necessary}).  Notably, this condition does not depend on the prior, the specific actions used, or the probabilities with which they are used. Thus, excessive randomization over actions leads to inefficiency.  With two players, an outcome can be pure; quasi-pure, meaning deterministic in all but one state where two actions are used; or mixed, meaning neither pure nor quasi-pure. Generically, only pure or quasi-pure outcomes can be efficient.

We then extend the analysis to infinite state spaces with a nonatomic prior. Generically, an   outcome is efficient only if it is pure almost surely (\Cref{thm:infinite}). Thus, while the finite state result limits the amount of randomization compatible with efficiency,  the infinite state result rules out randomization on any set of states with positive probability.

The key insight behind our main results is geometric and links ex-ante and ex-post efficiency. Ex-post efficiency in a given state requires the state-contingent outcome to maximize some weighted sum of the players' payoffs in that state. An outcome can be ex-post efficient in every state even if the supporting weights differ across states. Ex-ante efficiency is more demanding: there must exist a \emph{common}  vector of welfare weights that supports the state-contingent outcome across states (\Cref{prop:exanteexpost,prop:exanteexpost-infinite}).  Randomization across actions creates indifferences that this common vector must satisfy. Generically, once the outcome uses too many actions across states, such a common vector does not exist.

We first apply the infinite state result to the allocation problem without transfers studied by \citet{miralles2012cardinal}. Two indivisible objects need to be allocated among agents with privately known valuations. An outcome assigns to each valuation profile a lottery over allocations. The planner chooses a Bayesian incentive-compatible mechanism to maximize utilitarian welfare. When all agents prefer the same object, the optimal allocation rule randomizes over allocations to elicit the intensity of their preferences. However, this randomization cannot be supported by a common vector of welfare weights across such valuation profiles, implying that the induced outcome is Pareto inefficient (\Cref{prop:miralles}).  Thus, the randomization that helps satisfy incentive constraints comes at the cost of Pareto efficiency. 

We next consider canonical communication models: cheap talk and Bayesian persuasion. In cheap talk, the sender cannot commit to a signaling policy.  Any stochastic equilibrium outcome requires at least one player to be indifferent: if the sender mixes across messages, she must be indifferent between the induced receiver responses; if the receiver mixes across actions, he must be indifferent between those actions under the induced belief. For such an outcome to be efficient, the other player must be indifferent as well. Such simultaneous indifferences are non-generic, so only pure cheap talk outcomes can generically be efficient (\Cref{cheaptalk}).  When the sender's payoff is state-independent, we obtain a sharper characterization under mild assumptions: a cheap talk outcome is efficient if and only if the sender's most preferred action is induced with certainty (\Cref{cheaptalk-transparent}). Consequently, any cheap talk equilibrium in which communication affects the receiver's action is inefficient.

In Bayesian persuasion, we study a natural environment with one safe action and several risky actions. The sender prefers risky actions, whereas the receiver takes a risky action only when sufficiently confident about the corresponding state. To increase the probability of risky actions while satisfying obedience, the sender often randomizes over action recommendations. For a full-dimensional set of priors and a broad class of receiver preferences, the Bayesian persuasion outcome uses too many actions across states to satisfy the necessary condition in \Cref{necessary}, and is therefore generically inefficient (\Cref{BPfixedpref,BPfixedprior}).

\paragraph{Related Literature:} Given its importance, a broad literature studies the conditions under which outcomes are Pareto efficient. Early work examined the efficiency of Nash equilibria \citep{case1974class,Dubey1986,cohen1998cooperation} and conditions for implementing efficient outcomes under incomplete information \citep{holmstrom1983efficient,myersonbook1991}. Related work studies the tension between Bayesian incentive compatibility and welfare objectives in environments with cardinal preferences \citep{borgers2009efficient, ehlers2020continuity, kikuchi2024general}. Another line of work studies how efficient outcomes can be reached through learning or commitment procedures \citep{arieli2012average,pradelski2012learning,marden2014achieving,jindani2022learning,arieli2017sequential}. In a closely related paper, \citet{Rudov} analyze when a Nash equilibrium can be improved upon by a correlated equilibrium, using the convexity property of the set of correlated equilibria. Like us, they derive geometric conditions that restrict the extent of randomization. Our focus is different: rather than asking whether equilibria are efficient, whether efficient outcomes can be implemented, or how they can be reached, we characterize which feasible outcomes can lie on the ex-ante Pareto frontier. Evaluating efficiency relative to all feasible outcomes yields restrictions directly on the support of an outcome and, in particular, on the extent of randomization compatible with efficiency.

Our work also contributes to the literature on  strategic communication. In cheap talk \citep{crawford1982strategic, green2007two}, the sender's message is unverifiable, while in Bayesian persuasion \citep{kamenicagentzkow}, the sender commits to how the message is generated.\footnote{ There has  been some work on selecting equilibria in cheap talk games that are Pareto dominant; see \cite{crawford1982strategic} and \cite{antic2023equilibrium}.  We study efficiency relative to all feasible payoffs, not just equilibria.} The closest work within this literature is \citet{ichihashi2019limiting}, who analyzes how restrictions on the sender’s information in Bayesian persuasion affect the equilibrium outcome. In particular, he shows that the persuasion outcome is always efficient with binary actions, whereas our results demonstrate that efficiency may fail when there are more than two actions. \citet{rayo2010optimal} study disclosure rules that maximize a weighted sum of sender and receiver welfare. \citet{doval2024persuasion} analyze a welfare function over a heterogeneous population and characterize the Pareto frontier of  outcomes achievable through signaling policies.

Our results also relate to mechanism design without
transfers. Much of this literature focuses on deterministic mechanisms,
although this restriction is not without loss in general. In the absence of transfers, randomization can relax incentive constraints and improve the designer's objective
\citep{kovavc2009stochastic,miralles2012cardinal,niemeyer2024optimal}.
Existing work identifies conditions under which deterministic mechanisms are
nevertheless optimal or equivalent to stochastic mechanisms
\citep{strausz2006deterministic,kovavc2009stochastic,chen2019equivalence}. We instead study the welfare implications of randomization itself. In particular, Theorem \ref{thm:infinite} shows that, with a nonatomic prior, randomization is generically incompatible with Pareto efficiency.

The remainder of the paper is organized as follows. \Cref{sec:model} introduces the model and the notion of efficiency. \Cref{sec:mainresult} presents the main result for finite state spaces. \Cref{sec:infinite} develops the extension to infinite state spaces. \Cref{sec:applications} considers three applications: Bayesian allocation without transfers, cheap talk, and Bayesian persuasion. \Cref{sec:conclusion} concludes. All omitted results and proofs are in the Appendix.

\section{Model and Notion of Efficiency}
\label{sec:model}

We first present the model for a finite state space. The infinite state extension is given in Section \ref{sec:infinite}. There are $k\geq 2$ players and a finite set $A$ of feasible actions. For example, $A$ may be the Cartesian product $A=\prod_{i=1}^k A_i$, where $A_i$ is player $i$'s action set. Alternatively, $A$ may represent the set of feasible allocations.  The state of the world is drawn from a finite set $\Omega$ according to a common prior $p\in\mathrm{int}\,\Delta(\Omega)$. Player $i$'s payoff is $u_i:\Omega\times A\to\mathbb R$. Let
\[
u(\omega,a):=\bigl(u_1(\omega,a),\ldots,u_k(\omega,a)\bigr)\in\mathbb R^k
\]
denote the vector of players' payoffs.

An \emph{outcome} is a mapping $\mu\colon\Omega\to\Delta(A)$, where $\mu(a\mid\omega)$ is the probability that action $a$ is played in state $\omega$. When $A=\prod_{i=1}^k A_i$, one can think of an outcome as generated by a mediator who observes the realized state and, possibly at random, recommends an action to each player, which the players follow. These recommendations need not be incentive compatible; our efficiency benchmark is the set of all feasible outcomes.\footnote{Players may also have private information about the payoff-relevant state. Such information can affect incentive-compatibility constraints, but not the feasible payoff set.}

The payoff vector induced by outcome $\mu$ given prior $p$ is
\begin{equation}
\label{eq:finite-payoff-vector}
 u(\mu):=\sum_{\omega\in\Omega}p(\omega)\sum_{a\in A}\mu(a\mid\omega)u(\omega,a).
\end{equation}
The set of feasible payoff vectors given prior $p$ is
\begin{equation}
\label{eq:finite-feasible-set}
 F_p:=\{\bar u\in\mathbb R^k:\bar u=u(\mu)\text{ for some outcome }\mu\colon\Omega\to\Delta(A)\}.
\end{equation}
It is a convex polytope, and its extreme points are induced by pure outcomes.

\begin{definition}
Given a compact convex set of feasible payoffs $F\subseteq\mathbb R^k$, a vector $\bar u\in F$ is \textbf{efficient} if there is no $\bar v\in F$ such that $\bar v\geq \bar u$ and $\bar v_i>\bar u_i$ for at least one player $i$.
\end{definition}

In the finite model, every efficient feasible payoff vector maximizes a strictly positive weighted sum of the players’ payoffs.\footnote{This equivalence need not hold for arbitrary feasible payoff sets; see \cite{che2024near} for a related ``near'' weighted-sum characterization. In our finite model, it holds because $F_p$ is a convex polytope.}
Since each outcome induces a feasible payoff vector, we define efficiency of an outcome through the efficiency of its induced payoff vector.

\begin{definition}
An outcome $\mu\colon\Omega\to\Delta(A)$ is \textbf{efficient} if $u(\mu)$ is efficient with respect to $F_p$.
\end{definition}

Given an outcome $\mu$, we refer to $\mu(\omega)\in\Delta(A)$ and $u(\mu\mid\omega):=\sum_{a\in A}\mu(a\mid\omega)u(\omega,a)$ as the state-contingent outcome and the state-contingent payoff vector. Let
\begin{equation}
\label{eq:state-feasible-set}
 F_\omega:=\mathrm{co}\{u(\omega,a):a\in A\}
\end{equation}
denote the feasible payoff vectors in state $\omega$.  The feasible payoff set given a prior can be decomposed as the weighted Minkowski sum of the state-contingent feasible sets,
\begin{equation}
\label{eq:minkowski-decomposition}
 F_p=\sum_{\omega\in\Omega}p(\omega)F_\omega.
\end{equation}

Our notion of efficiency is ex ante: outcomes are evaluated before the state is realized.\footnote{There is no unique notion of efficiency under incomplete information. \citet{holmstrom1983efficient} distinguish several notions based on timing and feasibility. Our notions correspond to their ex-ante and ex-post classical efficiency.}  For comparison, we also define ex-post efficiency, which evaluates outcomes after the state is realized.

\begin{definition}
An outcome $\mu\colon\Omega\to\Delta(A)$ is \textbf{ex-post efficient} in state $\omega$ if $u(\mu\mid\omega)$ is efficient with respect to $F_\omega$.
\end{definition}

Ex-ante efficiency implies ex-post efficiency in every state, but the converse need not hold. The following proposition gives the geometric relation between the two notions.

\begin{proposition}
\label{prop:exanteexpost}
An outcome $\mu\colon \Omega \to \Delta(A)$ is efficient if and only if there exists a strictly positive vector
$n\in\mathbb{R}_{++}^k$ such that, for every state $\omega\in\Omega$,
\[
u(\mu\mid\omega)\in\arg\max_{x\in F_\omega} n^\top x.
\]
\end{proposition}

Intuitively, ex-ante efficiency requires a single strictly positive weight vector to support the outcome in every state.  If an outcome fails to maximize the corresponding weighted sum in some state, then some alternative action  gives a strictly higher weighted payoff in that state. Replacing the outcome in that state with this action  would therefore raise the overall weighted payoff, contradicting efficiency. In particular, efficiency depends only on the state-contingent support of the outcome, not on the probabilities assigned to actions within the support. By contrast, ex-post efficiency in state $\omega$ only requires that the state-contingent outcome be supported by some strictly positive weight vector $n_\omega$. Even if the outcome is ex-post efficient in every state, the supporting weights may vary across states, so the outcome need not be ex-ante efficient.

We illustrate the concepts of ex-ante and ex-post efficiency using the following example.

\begin{example} 
	\label{example}
	
	Consider a Bayesian persuasion problem with a sender and a receiver. The  state space is $\Omega=\{\omega_0,\omega_1 \}$ and the receiver's action set is $A=\{a_0,a_1,a_2,a_3, a_4 \}$.  The sender's and the receiver's payoffs are given by the following matrix:
	
	\begin{table}[h!]
		\centering
		\setlength{\extrarowheight}{2pt}
		\begin{tabular}{cc|c|c|c|c|c|}
			& \multicolumn{1}{c}{}  & \multicolumn{1}{c}{$a_0$} & \multicolumn{1}{c}{$a_1$}  & \multicolumn{1}{c}{$a_2$} & \multicolumn{1}{c}{$a_3$}  & \multicolumn{1}{c}{$a_4$} \\\cline{3-7}
			& $\omega_0$  & $(2,9)$ &  $(10,8)$ & $(0,6.4)$ &  $(3,4)$ & $(1,0)$ \\\cline{3-7}
			& $\omega_1$ & $(2,0)$ & $(10,4)$ & $(0,6.4)$ &  $(3,8)$  & $(1,9) $\\\cline{3-7}
		\end{tabular}
	\end{table}

We analyze the efficiency of the Bayesian persuasion outcomes for three priors, where $p=\mathbb{P}(\omega_1)$:
\begin{enumerate}[label=(\alph*)]
	\item $p=0.10$: $a_0$ and $a_1$ are taken in $\omega_0$, while $a_1$ is taken in $\omega_1$;
	\item $p=0.30$: $a_1$ is taken with certainty in both states;
	\item $p=0.70$: $a_1$ is taken in $\omega_0$, while $a_1$ and $a_4$ are taken in $\omega_1$.
\end{enumerate}

	\begin{figure}[htp!]
		\hspace*{-1.3cm}  
		\centering
		\begin{minipage}{0.40\textwidth}
			\centering
			\begin{tikzpicture}[scale=0.82]
				\begin{axis}[
					xmax=11.5,
					xmin=0,
					ymax=11.5,
					ymin=0,
					ylabel near ticks,
					xlabel near ticks,
					xlabel={$u_S$},
					ylabel={$u_R$},
					xtick={0,1,2,3,4,5,6,7,8,9,10},
					ytick={0,1,2,3,4,5,6,7,8,9,10},
					width=7cm,
					height=7cm
					]
					\addplot[thick, black, fill=red, fill opacity=0.50] coordinates {
						(0,6.4)
						(2,9) 
						(10,8) 
						(1,0)
						(0,6.4)
					};
					
					\addplot[thick, black, mark=*, only marks,fill opacity=0.50, nodes near coords={$a_{\coordindex}$}, nodes near coords align={west}] coordinates {
						(2,9) 
						(10,8) 
						(0,6.4)
						(3,4)
						(1,0)
					};
					\draw[dashed] (50/9,77/9) circle(.6cm);
					\draw[thick, fill=black] (50/9,77/9) circle(.1cm);

					\draw[thick,  black, ->] (50/9,77/9) -- (50/9+1/9*1.2,77/9+8/9*1.2)  node[right] {$n_{\omega_0}$};
					
				\end{axis}
			\end{tikzpicture}
		\end{minipage}%
		\begin{minipage}{0.40\textwidth}
			\centering
			\begin{tikzpicture}[scale=0.82]
				\begin{axis}[
					xmax=11.5,
					xmin=0,
					ymax=11.5,
					ymin=0,
					ylabel near ticks,
					xlabel near ticks,
					xlabel={$u_S$},
					ylabel={$u_R$},
					xtick={0,1,2,3,4,5,6,7,8,9,10},
					ytick={0,1,2,3,4,5,6,7,8,9,10},
					width=7cm,
					height=7cm
					]
					\addplot[thick, black, fill=blue, fill opacity=0.50] coordinates {
						(0,6.4)
						(2,0) 
						(10,4) 
						(3,8)
						(1,9)
						(0,6.4)
					};
					
					\addplot[thick, black, mark=*, only marks,fill opacity=0.50, nodes near coords={$a_{\coordindex}$}, nodes near coords align={west}] coordinates {
						(2,0) 
						(10,4) 
						(0,6.4)
						(3,8)
						(1,9)
						(2,0)
					};
					
					\draw[dashed] (10,4) circle(.6cm);
					\draw[thick, fill=black] (10,4) circle(.1cm);
					
					\draw[thick, black, ->] (10,4)  -- (10+1,4)   node[below left ] {$\underline{n}_{\omega_1}$}; 
					
					\draw[thick, black, ->] (10,4)  -- (10+1.5*4/11,4+1.5*7/11)  node[above ] {$\overline{n}_{\omega_1}$};

				\end{axis}
			\end{tikzpicture}
		\end{minipage}%
		\begin{minipage}{0.40\textwidth}
			\centering
			\begin{tikzpicture}[scale=0.82]
				\begin{axis}[
					xmax=11.5,
					xmin=0,
					ymax=11.5,
					ymin=0,
					ylabel near ticks,
					xlabel near ticks,
					xlabel={$u_S$},
					ylabel={$u_R$},
					xtick={0,1,2,3,4,5,6,7,8,9,10},
					ytick={0,1,2,3,4,5,6,7,8,9,10},
					width=7cm,
					height=7cm
					]

					\addplot[thick, black,  mark=*, fill=orange, fill opacity=0.50]
					coordinates{
						(0.0,6.4) (0.9,0.64) (1.1,0.0) (1.9,0.4)
						(10.0,7.6) (9.3,8.0) (9.1,8.1) (1.9,9.0)
						(0.1,6.66) (0.0,6.4)
					};
					\draw[dashed] (6,8.1) circle(.6cm);
					\draw[thick, fill=black] (6,8.1) circle(.1cm);
				\end{axis}
			\end{tikzpicture}
		\end{minipage}
		
		\par 
				\hspace*{-1.3cm}  
		\begin{minipage}{0.40\textwidth}
			\centering
			\begin{tikzpicture}[scale=0.82]
				\begin{axis}[
					xmax=11.5,
					xmin=0,
					ymax=11.5,
					ymin=0,
					ylabel near ticks,
					xlabel near ticks,
					xlabel={$u_S$},
					ylabel={$u_R$},
					xtick={0,1,2,3,4,5,6,7,8,9,10},
					ytick={0,1,2,3,4,5,6,7,8,9,10},
					width=7cm,
					height=7cm
					]
					\addplot[thick, black, fill=red, fill opacity=0.50] coordinates {
						(0,6.4)
						(2,9) 
						(10,8) 
						(1,0)
						(0,6.4)
					};
					
					\addplot[thick, black, mark=*, only marks,fill opacity=0.50, nodes near coords={$a_{\coordindex}$}, nodes near coords align={west}] coordinates {
						(2,9) 
						(10,8) 
						(0,6.4)
						(3,4)
						(1,0)
					};
					\draw[dashed] (10,8) circle(.6cm);
					\draw[thick, fill=black] (10,8) circle(.1cm);

					
				\end{axis}
			\end{tikzpicture}
		\end{minipage}%
		\begin{minipage}{0.40\textwidth}
			\centering
			\begin{tikzpicture}[scale=0.82]
				\begin{axis}[
					xmax=11.5,
					xmin=0,
					ymax=11.5,
					ymin=0,
					ylabel near ticks,
					xlabel near ticks,
					xlabel={$u_S$},
					ylabel={$u_R$},
					xtick={0,1,2,3,4,5,6,7,8,9,10},
					ytick={0,1,2,3,4,5,6,7,8,9,10},
					width=7cm,
					height=7cm
					]
					\addplot[thick, black, fill=blue, fill opacity=0.50] coordinates {
						(0,6.4)
						(2,0) 
						(10,4) 
						(3,8)
						(1,9)
						(0,6.4)
					};
					
					\addplot[thick, black, mark=*, only marks,fill opacity=0.50, nodes near coords={$a_{\coordindex}$}, nodes near coords align={west}] coordinates {
						(2,0) 
						(10,4) 
						(0,6.4)
						(3,8)
						(1,9)
						(2,0)
					};
					
					\draw[dashed] (10,4) circle(.6cm);
					\draw[thick, fill=black] (10,4) circle(.1cm);

					
				\end{axis}
			\end{tikzpicture}
		\end{minipage}%
		\begin{minipage}{0.40\textwidth}
			\centering
			\begin{tikzpicture}[scale=0.82]
				\begin{axis}[
					xmax=11.5,
					xmin=0,
					ymax=11.5,
					ymin=0,
					ylabel near ticks,
					xlabel near ticks,
					xlabel={$u_S$},
					ylabel={$u_R$},
					xtick={0,1,2,3,4,5,6,7,8,9,10},
					ytick={0,1,2,3,4,5,6,7,8,9,10},
					width=7cm,
					height=7cm
					]
					
					\addplot[thick, black, fill=orange, mark=*, fill opacity=0.50]
					coordinates{
						(0.0,6.4) (0.7,1.8) (1.3,0.0) (3.7,1.2)
						(10.0,6.8) (7.9,8.0) (7.3,8.3) (1.7,9.0)
						(0.3,6.9) (0.0,6.4)
					};
					
					\draw[dashed] (10,6.8) circle(.6cm);
					\draw[thick, fill=black] (10,6.8) circle(.1cm);
				\end{axis}
			\end{tikzpicture}
		\end{minipage}
		
		\par
				\hspace*{-1.3cm}  
		\begin{minipage}{0.40\textwidth}
			\centering
			\begin{tikzpicture}[scale=0.82]
				\begin{axis}[
					xmax=11.5,
					xmin=0,
					ymax=11.5,
					ymin=0,
					ylabel near ticks,
					xlabel near ticks,
					xlabel={$u_S$},
					ylabel={$u_R$},
					xtick={0,1,2,3,4,5,6,7,8,9,10},
					ytick={0,1,2,3,4,5,6,7,8,9,10},
					width=7cm,
					height=7cm
					]
					\addplot[thick, black, fill=red, fill opacity=0.50] coordinates {
						(0,6.4)
						(2,9) 
						(10,8) 
						(1,0)
						(0,6.4)
					};
					
					\addplot[thick, black, mark=*, only marks,fill opacity=0.50, nodes near coords={$a_{\coordindex}$}, nodes near coords align={west}] coordinates {
						(2,9) 
						(10,8) 
						(0,6.4)
						(3,4)
						(1,0)
					};
					\draw[dashed] (10,8) circle(.6cm);
					\draw[thick, fill=black] (10,8) circle(.1cm);

				\end{axis}
			\end{tikzpicture}
		\end{minipage}%
		\begin{minipage}{0.40\textwidth}
			\centering
			\begin{tikzpicture}[scale=0.82]
				\begin{axis}[
					xmax=11.5,
					xmin=0,
					ymax=11.5,
					ymin=0,
					ylabel near ticks,
					xlabel near ticks,
					xlabel={$u_S$},
					ylabel={$u_R$},
					xtick={0,1,2,3,4,5,6,7,8,9,10},
					ytick={0,1,2,3,4,5,6,7,8,9,10},
					width=7cm,
					height=7cm
					]
					\addplot[thick, black, fill=blue, fill opacity=0.50] coordinates {
						(0,6.4)
						(2,0) 
						(10,4) 
						(3,8)
						(1,9)
						(0,6.4)
					};
					
					\addplot[thick, black, mark=*, only marks,fill opacity=0.50, nodes near coords={$a_{\coordindex}$}, nodes near coords align={west}] coordinates {
						(2,0) 
						(10,4) 
						(0,6.4)
						(3,8)
						(1,9)
						(2,0)
					};
					\draw[dashed] (axis cs:{25/7},{53/7}) circle(.6cm);
\draw[thick, fill=black] (axis cs:{25/7},{53/7}) circle(.1cm);

				\end{axis}
			\end{tikzpicture}
		\end{minipage}%
		\begin{minipage}{0.40\textwidth}
			\centering
			\begin{tikzpicture}[scale=0.82]
				\begin{axis}[
					xmax=11.5,
					xmin=0,
					ymax=11.5,
					ymin=0,
					ylabel near ticks,
					xlabel near ticks,
					xlabel={$u_S$},
					ylabel={$u_R$},
					xtick={0,1,2,3,4,5,6,7,8,9,10},
					ytick={0,1,2,3,4,5,6,7,8,9,10},
					width=7cm,
					height=7cm
					]
					
					\addplot[thick, black, fill=orange, mark=*, fill opacity=0.50]
					coordinates{
						(0.0,6.4) (0.3,4.2) (1.7,0.0) (7.3,2.8)
						(10.0,5.2) (5.1,8.0) (3.7,8.7) (1.3,9.0)
						(0.7,8.1) (0.0,6.4)
					};

					\draw[dashed] (axis cs:5.5,7.7) circle(.6cm);
\draw[thick, fill=black] (axis cs:5.5,7.7) circle(.1cm);
				\end{axis}
			\end{tikzpicture}
		\end{minipage}
		\caption{Each row corresponds to a different prior:  (a)~$p=0.10$, (b)~$p=0.30$, and (c)~$p=0.70$.  Within each panel, the red region represents $F_{\omega_0}$, the blue region represents $F_{\omega_1}$, and the orange region represents $F_p$. The black circled node denotes the outcome.}
		\label{fig:efficiency}
	\end{figure}

	The feasible payoff sets for states $\omega_0$, $\omega_1$, and under  prior $p$ are represented by the red, blue, and orange regions in \Cref{fig:efficiency}, respectively. We find that for (a), the outcome is ex-post efficient in both states but is not supported by a common positive normal. The unique normal $n_{\omega_0}$ does not belong to the normal cone spanned by $\overline{n}_{\omega_1}$ and ${\underline{n}}_{\omega_1}$, i.e., $n_{\omega_0} \notin \mathrm{cone}\{\overline{n}_{\omega_1}, {\underline{n}}_{\omega_1}\}$. For (b), the outcome is ex-post efficient in both states and supported by a common positive normal, while for (c), the outcome is not ex-post efficient in state $\omega_1$. Overall, the outcome is efficient in (b) but not in (a) or (c) (see \Cref{fig:efficiency}). 
\end{example}

\section{Necessary condition for efficiency}
\label{sec:mainresult}

Our main result provides a necessary condition for efficiency based on the number of actions taken across states. Given an outcome $\mu$ and a state $\omega$, let $|\mu(\omega)|$ denote the size of the support of $\mu(\omega)$, that is, the number of actions taken with positive probability in that state.  Throughout this section, we say a property holds generically if it holds outside a subset of the payoff space  $\mathbb R^{k|\Omega||A|}$ with Lebesgue measure zero.

Generically, efficiency requires that the total number of actions taken across all states be strictly less than the sum of the number of players and states; excessive randomization therefore leads to inefficiency. Efficiency requires a common strictly positive weight vector that supports the outcome in every state. Generically, when too many actions are taken across states, no such vector exists.

\begin{theorem}
	\label{necessary}
	Generically,  an outcome $\mu\colon\Omega \to \Delta(A)$ is efficient only if
	\begin{equation} \label{efficiency:number}
		\sum_{\omega \in \Omega} |\mu(\omega)| < k + |\Omega| .
	\end{equation}
\end{theorem}

\begin{proof}
	
    By \Cref{prop:exanteexpost}, if $\mu$ is efficient, there exists a common
strictly positive weight vector $n\in\mathbb R^k_{++}$ that supports the
outcome in every state. Fix a state $\omega$ and let $a_1,\dots,a_{|\mu(\omega)|}$ denote the pure actions used by $\mu$ in that state. Since each of these actions must maximize the same weighted payoff,
\[
	n^\top\bigl(u(\omega,a_i)-u(\omega,a_1)\bigr)=0
	\quad\text{for }i=2,\dots,|\mu(\omega)|.
\]

Collect these indifference conditions across states in a matrix
$D_\mu$. It has
\[
r=\sum_{\omega\in\Omega}(|\mu(\omega)|-1)
\]
rows, and the indifference conditions across all states can be written as $D_\mu n=0$.

Fix the supports of the outcome across states. If $r\geq k$, then, except on a measure-zero set of payoff profiles, at least $k$ of the associated payoff-difference vectors are linearly independent. Hence, $D_\mu$ has rank $k$, so $D_\mu n=0$ implies $n=0$. Since $A$ and $\Omega$ are finite, there are only finitely many possible supports across states; taking the union of such sets preserves measure zero.  

Thus, if
\[
\sum_{\omega\in\Omega} |\mu(\omega)| \geq k+|\Omega|,
\]
then $r\geq k$. Generically, no strictly positive common supporting vector exists, so $\mu$ is inefficient.
\end{proof}

Notably, the bound depends only on the number of actions taken across states, not on the interior prior, the specific actions  taken, or the probabilities with which they are taken. It also implies the state-wise bound $|\mu(\omega)|\leq k$ for all $\omega\in\Omega$. Generically, the same bound is necessary even for ex-post efficiency in a given state.

For two-player games, outcomes can be pure (a deterministic action in every state), quasi-pure (a deterministic action in all but one state, where exactly two actions are taken), or mixed (any other case). Generically, only pure or quasi-pure outcomes can be efficient.

For two-state games with $k$ players, the restriction becomes particularly strong under independent randomization. Suppose $k_i$ players independently randomize over at least two actions in state $\omega_i$ for $i=0,1$. Then $|\mu(\omega_i)|\geq 2^{k_i}$, so generically efficiency requires
\[
2^{k_0}+2^{k_1}< k+2.
\]
Hence, as the number of players grows, only a vanishing fraction of players can independently randomize. The same conclusion holds more generally for any fixed number of states.

The necessary condition in \Cref{necessary} is far from sufficient. For instance, in \Cref{example}, the outcome is quasi-pure and satisfies the bound in cases (a) and (c), yet is inefficient in both cases.  Even if the outcome is pure and ex-post efficient in every state, the overall outcome can still be inefficient. Efficiency is guaranteed, for example, when the outcome induces one player’s unique most preferred action in every state. In \Cref{example}, this occurs when the sender fully reveals the state and the receiver takes his optimal action in each state, or when the sender’s preferred action is chosen in every state, as in case (b).

A complete characterization of efficiency is provided in \Cref{propdeviations} in the Appendix. Although we take an ex-ante view, efficiency can be assessed by examining the payoff changes generated by deviations in each state. An outcome is efficient if and only if no nonnegative combination of these statewise deviations leads to a Pareto improvement. This condition depends only on the state-contingent support of the outcome, not on the interior prior or the randomization probabilities.

\section{Efficiency in Infinite State Spaces}
\label{sec:infinite}
In this section, we extend the analysis to infinite state spaces with a nonatomic prior. The main insight remains the same: excessive randomization leads to inefficiency. In fact,  the conclusion becomes sharper: generically, every efficient outcome is  pure almost surely.

Let $\Omega$ be a compact subset of $\mathbb R^d$, and let $p$ denote the nonatomic
common prior, a Borel probability measure on $\Omega$.  The set of actions $A$ remains finite. Let $u$ denote the payoff profile, where $u(\omega,a)\in\mathbb R^k$ is the vector of players' payoffs from action  $a$ in state $\omega$. For each $a\in A$, assume $\omega\mapsto u(\omega,a)$ is measurable and bounded. An outcome is a measurable mapping $\mu\colon\Omega\to\Delta(A)$, and it induces the payoff vector
\begin{equation}
\label{eq:infinite-payoff-vector}
 u(\mu):=\int_\Omega \sum_{a\in A}\mu(a\mid\omega)u(\omega,a)\,dp(\omega).
\end{equation}

 With infinite states, the feasible payoff set given prior $p$ can be written as the Aumann integral of the state-contingent feasible sets\footnote{For a measurable correspondence
\(\omega\mapsto F_\omega\subseteq\mathbb R^k\), its Aumann integral is
\[
\int_\Omega F_\omega\,dp(\omega)
:=
\left\{
\int_\Omega f(\omega)\,dp(\omega):
f \text{ is a measurable integrable selection of } F_\omega
\right\},
\]
where \(f\) is a selection if \(f(\omega)\in F_\omega\) for \(p\)-almost every
\(\omega\). In our setting, integrability follows from the boundedness of payoffs.}
\begin{equation}
\label{eq:aumann-integral}
 F_p=\int_\Omega F_\omega\,dp(\omega).
\end{equation}
The set $F_p$ is compact and convex.

As in the finite-state case, efficiency requires a common supporting weight vector in almost every state, although here the vector need only be nonzero and nonnegative.

\begin{proposition}
\label{prop:exanteexpost-infinite}
If an outcome $\mu\colon\Omega\to\Delta(A)$ is efficient, then there exists a nonzero nonnegative vector  $n\in\mathbb R^k_+\setminus\{0\}$ such that, for $p$-almost every state $\omega$,
\[
u(\mu\mid\omega)\in\arg\max_{x\in F_\omega} n^\top x.
\]
\end{proposition}

 Since the feasible payoff set is obtained by integrating the state-contingent feasible sets $F_\omega$, a weight vector supporting an ex-ante efficient payoff must support the outcome in almost every state. Otherwise, there would be a  set of states with positive probability in which the outcome fails to maximize the weighted payoff. Changing the outcome on that set would increase the overall weighted payoff, contradicting the fact that the ex-ante payoff is supported by this weight vector.

Let $\mathcal U$ denote the space of bounded measurable payoff profiles $u\colon\Omega\times A\to\mathbb R^k$, endowed with the supremum norm. Because $\mathcal U$ is infinite dimensional, there is no natural analogue of
Lebesgue measure on the entire space. We therefore use prevalence as our measure-theoretic notion of genericity; see \citet{hunt1992prevalence} and \citet{anderson2001genericity}.\footnote{Formal definitions of shyness, prevalence, and finite prevalence are provided in \Cref{app:no-ties-prevalence}.} Throughout this section, a property holds generically if it holds on a
prevalent subset of $\mathcal U$. 

We now state our main result for infinite state spaces.

\begin{theorem}
\label{thm:infinite}
Generically, an outcome $\mu\colon\Omega\to\Delta(A)$ is efficient only if it is pure $p$-almost surely.
\end{theorem}

\begin{proof}
The key argument is that efficiency requires a common vector of welfare weights to support the outcome in almost every state. Thus, if an efficient outcome randomizes between two distinct actions on a set of states with positive probability, these pure actions must have the same weighted payoff on that set.  We show that, given a nonatomic prior, payoff profiles exhibiting such weighted ties on a set of states with positive probability form a finitely shy set.

Formally, let
\[
\mathcal B
:=
\left\{
u\in\mathcal U\colon
\text{some efficient outcome is not pure $p$-almost surely}
\right\}
\]
denote the set of payoff profiles that admit at least one efficient outcome which randomizes on a set of states of positive probability.

Fix $u\in\mathcal B$, and let $\mu$ be an efficient outcome that is not pure $p$-almost surely. By
\Cref{prop:exanteexpost-infinite}, there exists
$n\in\mathbb R^k_+\setminus\{0\}$ such that, for $p$-almost every state $\omega$, every action  in the support of $\mu(\cdot\mid\omega)$ maximizes $n^\top u(\omega,\cdot)$ over $A$.

Since $\mu$ is not pure $p$-almost surely and $A$ is finite, there exist distinct actions $a,b\in A$ that both belong to the support of $\mu(\cdot\mid\omega)$ on a set with positive probability.  For almost every state in this set, both actions maximize the same weighted payoff. Hence,

\[
p\Bigl(
\bigl\{\omega\in\Omega:
n^\top u(\omega,a)=n^\top u(\omega,b)
\bigr\}
\Bigr)>0.
\]

Let $E_{a,b}\subseteq\mathcal U$ denote the set of payoff profiles for which the above weighted tie condition holds for some
$n\in\mathbb R_+^k\setminus\{0\}$. Thus, every $u\in\mathcal B$ belongs to $E_{a,b}$ for some distinct $a,b\in A$, and hence
\[
\mathcal B\subseteq\bigcup_{a\neq b}E_{a,b}.
\]

By \Cref{prop:infinite-prevalence} in the Appendix, each $E_{a,b}$ is finitely shy. Since $A$ is finite, $\mathcal B$ is also finitely shy. Hence, its complement is finitely prevalent, and therefore prevalent. Thus, generically, every efficient outcome is pure $p$-almost surely.
\end{proof}

Unlike the finite state result, which limits the extent of randomization, \Cref{thm:infinite} rules out randomization on any set of states with positive probability.

The result is relevant for mechanism design without transfers. The state space $\Omega$ may represent values, costs, qualities, or a profile of privately observed types, while $A$ is the finite set of feasible allocations—for example, which agent receives an object. Thus, with a nonatomic prior, stochastic allocation rules are generically incompatible with Pareto efficiency, even when randomization is useful for incentive compatibility \citep{kovavc2009stochastic,miralles2012cardinal}.\footnote{The genericity statement is with respect to the unrestricted space of payoff profiles. Many applications impose additional structure on payoffs; for example, in mechanism design an agent's payoff may depend only on her own type and the allocation. In such cases, one may study genericity relative to the restricted class or apply Proposition~\ref{prop:exanteexpost-infinite} directly, as in Section~\ref{sec:allocation}.}

\section{Applications}
\label{sec:applications}

We apply our results to three canonical environments:  Bayesian allocation without transfers, cheap talk, and Bayesian persuasion. The first application
uses the infinite state result, while the latter two use the finite state analysis.

\subsection{Bayesian allocation without transfers}
\label{sec:allocation}

We first apply the infinite state result to an allocation problem without monetary transfers, studied by \citet{miralles2012cardinal}. Two ex-ante identical indivisible objects are allocated among $k$ agents. Agent $i$ privately observes her valuation vector $\omega_i=(\omega_{i1},\omega_{i2})\in[0,1]^2$, drawn independently from a common distribution $F$ with a symmetric density. The state is the valuation profile of all agents $\omega=(\omega_1,\ldots,\omega_k)$. A deterministic allocation is $a=(a_1,a_2)$, where $a_j$ indicates which agent receives object $j$, allowing an object to remain unassigned. Let $A$ denote the finite set of such allocations. Valuations are additive across objects, so agent $i$'s payoff is
\[
u_i(\omega,a)=
\omega_{i1}\mathbf 1\{a_1=i\}
+
\omega_{i2}\mathbf 1\{a_2=i\}.
\]

Agents report their valuations to the planner, and a mechanism specifies a lottery over allocations for every profile of reports. Bayesian incentive compatibility requires that truthful reporting maximize each agent's interim expected utility, given that the other agents report truthfully. The planner chooses a Bayesian incentive-compatible mechanism to maximize the sum of the agents' ex-ante expected payoffs. Under truthful reporting, we identify the allocation rule with the induced outcome
$\mu\colon\Omega\to\Delta(A)$.

\citet{miralles2012cardinal} characterizes the optimal allocation rule when $F$ is regular and satisfies a monotonicity condition.\footnote{Roughly, regularity ensures that assigning an object to an agent who prefers it generates more welfare than assigning it to an agent who moderately prefers the other object. The monotonicity condition ensures that, when all agents prefer the same object, the other object is assigned to the agent with the most moderate preferences. See \citet{miralles2012cardinal} for formal definitions.} We maintain these assumptions throughout this subsection. The optimal allocation rule is implemented through a combination of lotteries and auctions.


For illustration, consider the case with two agents. If they prefer different objects, each receives her preferred object. If they prefer the same object, each is initially given probability $1/2$ of receiving it, while the other object is auctioned. The winner pays her bid in probability units of the preferred object, which are transferred to the loser. Hence, when both agents prefer object $1$ and agent $1$ wins the auction for object $2$, the mechanism randomizes between giving both objects to agent $1$ and giving object $2$ to agent $1$ and object $1$ to agent $2$. Notably, each allocation in this randomization is ex-post efficient. With more than two agents, the auction is supplemented by an insurance scheme that relaxes the constraint imposed by agents' initial probability endowments and ensures feasibility.

Randomization plays a central role in the optimal allocation rule. Allocation probabilities serve as a substitute for monetary transfers and allow the planner to exploit the intensity with which agents prefer one object to the other.  \Cref{thm:infinite} suggests that such randomization is generically incompatible with efficiency. In this environment, we can go further and show that the optimal allocation rule is inefficient.

\begin{proposition}
\label{prop:miralles}
Suppose $F$ is regular and satisfies the monotonicity condition. Then, for every $k\geq2$, the outcome induced by the optimal allocation rule is inefficient.
\end{proposition}

To see why, consider the case with two agents and recall \Cref{prop:exanteexpost-infinite}. If the allocation rule were efficient,  there would exist a common  vector
$n=(n_1,n_2)\in\mathbb R_+^2\setminus\{0\}$ that supports the outcome at almost every valuation profile. In particular, it must support the outcome at almost every valuation profile where the allocation rule randomizes. The weighted welfare from giving both objects to agent $1$ is $n_1(\omega_{11}+\omega_{12})$, whereas that from giving object $2$ to agent $1$ and object $1$ to agent $2$ is
$n_1\omega_{12}+n_2\omega_{21}$. Since both allocations are in the support of the allocation rule on a set of positive probability, efficiency would require
\[
n_1\omega_{11}=n_2\omega_{21}
\]
for almost every valuation profile in this set. But since $F$ has a density, for any fixed $n$ this equality holds with probability zero. Thus, no common vector can support this allocation rule.

Although the mechanism maximizes utilitarian welfare subject to incentive compatibility, the induced outcome is Pareto inefficient. Thus, the randomization used to elicit private information forces the welfare-maximizing allocation rule away from the underlying Pareto frontier. 

\subsection{Cheap talk}

We  consider the cheap talk model \citep{crawford1982strategic,green2007two}, in which an informed sender communicates with an uninformed receiver without commitment.

The state $\omega\in\Omega$ is drawn according to the common prior $p$ and observed by the sender, who sends a message $m\in M$. After observing the message, the receiver chooses an action $a\in A$. The sender's and receiver's payoffs are $u_S(\omega,a)$ and $u_R(\omega,a)$, respectively. The sender's strategy $\sigma\colon\Omega\to\Delta(M)$ and the receiver's strategy $\tau\colon M\to\Delta(A)$ induce the outcome
\[
\mu(a\mid\omega)
=
\sum_{m\in M}\sigma(m\mid\omega)\tau(a\mid m).
\]

We refer to any outcome induced by a Perfect Bayesian equilibrium as a \emph{cheap talk outcome}. Cheap talk generally admits multiple equilibria, including a \emph{babbling equilibrium}, which always exists: no information is conveyed, and the receiver plays a best response to the prior.\footnote{Unlike in Bayesian persuasion, we cannot restrict attention to direct signaling policies, since some cheap talk equilibria require the receiver to randomize between actions after a given message.}

For cheap talk, \Cref{necessary} can be strengthened: generically, any stochastic outcome, including a quasi-pure one, is inefficient.

\begin{proposition} \label{cheaptalk}
	Generically, a cheap talk outcome is  efficient only if it is pure.
\end{proposition}

In any stochastic cheap talk outcome, either the sender mixes across messages in some state or the receiver mixes after an on-path message. In the first case, the sender must be indifferent between the induced lotteries over actions in that state; in the second, the receiver must be indifferent between the pure actions played at the induced posterior. Efficiency requires the other player to be indifferent as well. However, these simultaneous indifferences are non-generic. Thus, only pure cheap talk outcomes can be efficient.

Recently, \citet{kamenica2024commitment} show that, generically, if the sender must randomize in all of his preferred cheap talk equilibria, then he values commitment. Together with our result, this shows that a pure outcome is necessary for efficiency and, among the sender's preferred equilibria, for commitment to have no value.

For any receiver belief $q\in\Delta(\Omega)$, let
\[
A^*(q):=\arg\max_{a\in A}\mathbb{E}_q[u_R(\omega,a)]
\]
denote the receiver's set of best responses, and define the sender's value function by
\[
V(q):=\max_{a\in A^*(q)}\mathbb{E}_q[u_S(\omega,a)].
\]
Thus, $V(q)$ is the sender's expected payoff when the receiver selects the sender's preferred action among his best responses.

Next, we consider the case in which the sender's payoff is state-independent.  Using a belief-based approach, \citet{lipnowski2020cheap} show that the sender's preferred cheap talk equilibrium payoff is given by the quasiconcave envelope of $V$, evaluated at the prior $p$. We denote this envelope by $\mathrm{Quasicav}\,V$ and illustrate it graphically for \Cref{example} in \Cref{fig:cheaptalk}.

\begin{figure}[htp!]
\centering

\begin{minipage}[t]{0.49\textwidth}
\centering
\begin{tikzpicture}
\begin{axis}[
    width=\linewidth,
    height=0.72\linewidth,
    xlabel={$q$},
    ylabel={$u_S$},
    axis x line=bottom,
    axis y line=left,
    xmin=0,
    xmax=1,
    ymin=0,
    ymax=11.5,
    xtick={0,0.2,0.4,0.6,0.8,1},
    ytick={0,2,4,6,8,10},
legend style={
    at={(0.98,0.98)},
    anchor=north east,
    draw=black,
    fill=white,
    rounded corners=0pt,
    inner sep=2pt,
    font=\small
}
]

\addplot[MyBlue,ultra thick,forget plot]
    coordinates {(0,2) (0.2,2)};
\addplot[MyBlue,ultra thick,forget plot]
    coordinates {(0.2,10) (0.4,10)};
\addplot[MyBlue,ultra thick,forget plot]
    coordinates {(0.4,0) (0.6,0)};
\addplot[MyBlue,ultra thick,forget plot]
    coordinates {(0.6,3) (0.8,3)};
\addplot[MyBlue,ultra thick,forget plot]
    coordinates {(0.8,1) (1,1)};

\addplot[MyBlue,dotted,line width=1pt,forget plot]
    coordinates {(0.2,2) (0.2,10)};
\addplot[MyBlue,dotted,line width=1pt,forget plot]
    coordinates {(0.4,0) (0.4,10)};
\addplot[MyBlue,dotted,line width=1pt,forget plot]
    coordinates {(0.6,0) (0.6,3)};
\addplot[MyBlue,dotted,line width=1pt,forget plot]
    coordinates {(0.8,1) (0.8,3)};

\addplot[MyRed,dotted,line width=2pt,forget plot]
    coordinates {(0,2) (0.2,2)};
\addplot[MyRed,dotted,line width=2pt,forget plot]
    coordinates {(0.2,10) (0.4,10)};
\addplot[MyRed,dotted,line width=2pt,forget plot]
    coordinates {(0.4,3) (0.8,3)};
\addplot[MyRed,dotted,line width=2pt,forget plot]
    coordinates {(0.8,1) (1,1)};

\addlegendimage{MyBlue,ultra thick}
\addlegendentry{$V$}
\addlegendimage{MyRed,dotted,line width=2pt}
\addlegendentry{$\mathrm{Quasicav}\,V$}

\end{axis}
\end{tikzpicture}

\textbf{(a) Cheap talk}
\end{minipage}
\hfill
\begin{minipage}[t]{0.49\textwidth}
\centering
\begin{tikzpicture}
\begin{axis}[
    width=\linewidth,
    height=0.72\linewidth,
    xlabel={$q$},
    ylabel={$u_S$},
    axis x line=bottom,
    axis y line=left,
    xmin=0,
    xmax=1,
    ymin=0,
    ymax=11.5,
    xtick={0,0.2,0.4,0.6,0.8,1},
    ytick={0,2,4,6,8,10},
legend style={
    at={(0.98,0.98)},
    anchor=north east,
    draw=black,
    fill=white,
    rounded corners=0pt,
    inner sep=2pt,
    font=\small
}
]

\addplot[MyBlue,ultra thick,forget plot]
    coordinates {(0,2) (0.2,2)};
\addplot[MyBlue,ultra thick,forget plot]
    coordinates {(0.2,10) (0.4,10)};
\addplot[MyBlue,ultra thick,forget plot]
    coordinates {(0.4,0) (0.6,0)};
\addplot[MyBlue,ultra thick,forget plot]
    coordinates {(0.6,3) (0.8,3)};
\addplot[MyBlue,ultra thick,forget plot]
    coordinates {(0.8,1) (1,1)};

\addplot[MyBlue,dotted,line width=1pt,forget plot]
    coordinates {(0.2,2) (0.2,10)};
\addplot[MyBlue,dotted,line width=1pt,forget plot]
    coordinates {(0.4,0) (0.4,10)};
\addplot[MyBlue,dotted,line width=1pt,forget plot]
    coordinates {(0.6,0) (0.6,3)};
\addplot[MyBlue,dotted,line width=1pt,forget plot]
    coordinates {(0.8,1) (0.8,3)};

\addplot[MyRed,dotted,line width=2pt,forget plot]
    coordinates {(0,2) (0.2,10) (0.4,10) (1,1)};

\addlegendimage{MyBlue,ultra thick}
\addlegendentry{$V$}
\addlegendimage{MyRed,dotted,line width=2pt}
\addlegendentry{$\mathrm{Cav}\,V$}

\end{axis}
\end{tikzpicture}

\textbf{(b) Bayesian persuasion}
\end{minipage}

\caption{The sender's value function $V$, its quasiconcave envelope under cheap talk, and its concave envelope under Bayesian persuasion. Vertical dotted lines indicate jumps in $V$.}
\label{fig:cheaptalk}
\end{figure}

For the next result, suppose the action set is non-redundant:   every action is a best response at some  belief: $A=\bigcup_{q\in \Delta(\Omega)} A^*(q)$.  To  rule out non-generic cases,  assume also that the sender’s payoffs are distinct across actions: $u_S(a) \neq u_S(b)$ whenever $a \neq b$.    Let $a^*:=\arg\max_{a\in A}u_S(a)$ denote the sender's most preferred action. 

We now characterize efficient cheap talk outcomes when the sender’s payoff is state-independent.

\begin{proposition}\label{cheaptalk-transparent}
	Suppose the action set is non-redundant and the sender’s payoffs are distinct across actions. If the sender's payoff is state-independent, a cheap talk outcome is efficient if and only if the sender's most preferred action $a^*$ is induced with certainty.
\end{proposition}

If the receiver mixes after an on-path message, consider two actions $a_1$ and $a_2$ in the support, with $u_S(a_1)>u_S(a_2)$. Since both are the receiver's best responses at the induced posterior, there exists a state in its support in which the receiver weakly prefers $a_1$ to $a_2$. Shifting probability from $a_2$ to $a_1$ in that state is therefore a Pareto improvement.  If instead the receiver plays a pure action after every on-path message, sender incentives imply that all such messages induce the same action $a$. If $a\neq a^*$, non-redundancy implies that the receiver weakly prefers $a^*$ to $a$ in some state, so replacing $a$ by $a^*$ in that state is a Pareto improvement. Conversely, if $a^*$ is induced with certainty, the sender obtains his unique highest feasible payoff, and the outcome is efficient.

In \Cref{example}, the informative equilibrium at prior $p=0.5$ induces posteriors $q=0.4$ and $q=0.6$. At $q=0.4$, the receiver mixes between $a_1$ and $a_2$ to make the sender indifferent between the two messages. However, in state $\omega_0$, both players prefer $a_1$ to $a_2$, so the outcome is not ex-post efficient. Likewise, at $p=0.7$, the babbling equilibrium induces $a_3$, although both players prefer $a^*=a_1$ to $a_3$ in state $\omega_0$. By contrast, at $p=0.3$, the babbling equilibrium induces $a^*=a_1$ and is efficient.

Thus, under these assumptions, any cheap talk outcome in which communication affects the receiver's action is inefficient.

\subsection{Bayesian persuasion}
\label{sec:BP}

We now turn to Bayesian persuasion \citep{kamenicagentzkow}, in which the sender commits to a signaling policy before the state is realized.

The sender chooses a message space $M$ and a signaling policy $\sigma\colon\Omega\to\Delta(M)$. After the state $\omega\in\Omega$ is drawn according to the prior $p$, a message is sent according to $\sigma$, and the receiver chooses an action $a\in A$. As is standard, ties are broken in favor of the sender. Without loss, we restrict attention to direct signaling policies $\mu\colon\Omega\to\Delta(A)$, where messages are action recommendations. Given prior $p$, let $\mu_p^*$ denote the Bayesian persuasion outcome.\footnote{Generically, the Bayesian persuasion outcome is unique.} The \textbf{Bayesian persuasion (BP) outcome} $\mu_p^*\colon\Omega\to\Delta(A)$ solves
\begin{equation}
\max_{\mu\colon\Omega\to\Delta(A)}
\sum_{\omega\in\Omega}p(\omega)
\sum_{a\in A}\mu(a\mid\omega)u_S(\omega,a)
\end{equation}
subject to
\begin{equation}
\sum_{\omega\in\Omega}
p(\omega)\mu(a\mid\omega)
\bigl(u_R(\omega,a)-u_R(\omega,b)\bigr)
\geq 0
\qquad
\forall a,b\in A.
\label{obedience}
\end{equation}

\Cref{obedience} is the receiver's obedience condition: upon receiving recommendation $a$, the receiver prefers following it to choosing any other action $b$.

\cite{kamenicagentzkow} characterize the BP outcome  using the concavification approach \citep{aumann1995repeated}. Let $ \mathrm{Cav} \; V\colon \Delta (\Omega) \to \mathbb{R}$ denote the concave envelope of the value function $V$. The sender's expected payoff in the equilibrium outcome  is given by the evaluation of the concave envelope at the prior:  $\mathbb{E}_{\mu^{*}_p}[u_S(\omega,a)]=\mathrm{Cav}  \; V (p)$.

The value function and its concave envelope for \Cref{example} are depicted in \Cref{fig:cheaptalk}. The outcomes in \Cref{example} correspond  to the BP outcomes for the respective priors: (a) $p=0.10$, (b) $p=0.30$, and (c) $p=0.70$. As seen earlier, the BP  outcome is efficient for case (b), not for cases (a) or (c). Moreover, within the relative interior of each affine region, the BP outcome has the same support of actions across all priors in that region. Hence, it is efficient for $p\in[0.2,0.4]$ and inefficient for $p\in(0,0.2)\cup(0.4,1)$.

\paragraph{Safe and risky actions}

We next study an environment with one safe action and several risky actions. The receiver takes risky action $a_i$ only when sufficiently confident that state $\omega_i$ has occurred; otherwise, he chooses the safe action $a_0$. The sender prefers every risky action to the safe one. A natural example is a seller who prefers making any sale to no sale, while the buyer purchases product $i$ only when sufficiently confident that it is the best choice and otherwise does not buy.

The action set is $A=\{a_0,a_1,\ldots,a_n\}$, where $a_0$ is safe and $a_1,\ldots,a_n$ are risky. The state space is $\Omega=\{\omega_0,\omega_1,\ldots,\omega_n\}$, with $a_i$ optimal for the receiver in state $\omega_i$. The sender’s payoff is state independent: he gets $u_S(a_i)>0$ from every risky action and $u_S(a_0)=0$ from the safe action. We assume that the prior $p\in\operatorname{int}\Delta(\Omega)$ lies in the region where $a_0$ is optimal for the receiver. To describe the receiver’s preferences more generally, we specify the set of beliefs under which each action is optimal.

Recall that $A^*(q)$ denotes the set of receiver's best responses given belief $q$. Let \(C_i\subseteq\Delta(\Omega)\) denote the set of beliefs at which action \(a_i\) is a receiver best response:
\[
C_i := \{ q \in \Delta(\Omega) \colon a_i \in A^*(q) \}.
\]
The collection $\mathcal{P}=\{C_0,C_1,\ldots,C_n\}$ covers the belief space $\Delta(\Omega)$. We assume:
\begin{enumerate}
\item For each $i=0,\ldots,n$, we have $\omega_i \in C_i$, and $C_i$ contains an open neighborhood of $\omega_i$.

	\item For any distinct  $i \neq j$ with $i,j \neq 0$, we have $C_i \cap C_j = \emptyset$.
\end{enumerate}

Condition 1 ensures that action $a_i$ is optimal for beliefs sufficiently close to $\omega_i$, while Condition 2 ensures that two different risky actions are never optimal at the same belief. Hence, for each risky action $a_i$, the relevant boundary is where the receiver is indifferent between $a_i$ and the safe action $a_0$. Along each edge joining $\omega_i$ to $\omega_j$, this indifference occurs at a unique belief $o_{ij} \in C_i$, for $j\neq i$.\footnote{As we break ties in favor of the sender, if the belief $o_{ij}$ is induced, the receiver takes action $a_i$, since $u_S(a_i)>u_S(a_0)$.}

To characterize the BP outcome $\mu^{*}_p$, we follow \cite{lipnowski2017simplifying} and  restrict attention to policies supported on the finite set of \emph{outer points} $\mathrm{Out}(\mathcal{P})$, defined as
\begin{align*}
    \mathrm{Out}(\mathcal P)
    &:=
    \left\{
    q\in\Delta(\Omega):
    q\in\mathrm{ext}(C_i)
    \text{ whenever }q\in C_i\in\mathcal P
    \right\} \\
    &=
    \{\omega_0,\ldots,\omega_n\}
    \cup
    \{o_{ij}:i=1,\ldots,n,\ j\neq i\}.
\end{align*}

As shown by \citet{lipnowski2017simplifying}, an optimal policy can be supported on $\operatorname{Out}(\mathcal P)$, and its supporting posteriors can be chosen to be affinely independent. Since the BP outcome is generically unique, we may therefore focus on an optimal policy supported on the outer points.

\begin{figure}[!htbp] 
	\centering
	\begin{tikzpicture}
		\draw[fill=white, xscale=2/sqrt(3),xslant=.5]
		(3,0)-|(0,3)--cycle;
		\draw[fill=red, fill opacity=0.50, xscale=2/sqrt(3),xslant=.5] 
		(3,0)-- (2,0)--(2,1)--(3,0);
		\draw[fill=blue, fill opacity=0.50, xscale=2/sqrt(3),xslant=.5]  
		(0,3)-- (0,2) -- (1,2) -- (0,3);
		
		\draw[fill=orange, fill opacity=0.50,  xscale=2/sqrt(3),xslant=.5] 
		(0,0) -- (1.32,0.68) -- (0.68,1.32) -- (0,0);
		
		\filldraw[black, xscale=2/sqrt(3),xslant=.5] (2.2,0.65) circle (0pt) node[anchor=north]{$C_1$};
		\filldraw[black, xscale=2/sqrt(3),xslant=.5] (0.25,2.6) circle (0pt) node[anchor=north]{$C_2$};
		\filldraw[black, xscale=2/sqrt(3),xslant=.5] (0,0) circle (0pt) node[anchor=east]{$\omega_0$};
		\filldraw[black ,xscale=2/sqrt(3),xslant=.5] (3,0) circle (0pt) node[anchor=west]{$\omega_1$};
		\filldraw[black, xscale=2/sqrt(3),xslant=.5] (0,3) circle (0pt) node[anchor=south]{$\omega_2$};

		\draw[dotted,xscale=2/sqrt(3),xslant=.5] (0,0) -- (2,1) -- (2.4,1.2);
		
		\draw[dotted,xscale=2/sqrt(3),xslant=.5] (0,0) -- (1,2) --(1.2,2.4);
		
		\draw[dotted,xscale=2/sqrt(3),xslant=.5] (0,2) -- (2,0) -- (2.5,-0.5);
		
		\filldraw[black, xscale=2/sqrt(3),xslant=.5] (1,2) circle (0pt) node[anchor=west]{$o_{21}$};
		
		\filldraw[black, xscale=2/sqrt(3),xslant=.5] (2,1) circle (0pt) node[anchor= west]{$o_{12}$};

		\filldraw[black, xscale=2/sqrt(3),xslant=.5] (0,2) circle (0pt) node[anchor=east]{$o_{20}$};
		
		\filldraw[black, xscale=2/sqrt(3),xslant=.5] (2,0) circle (0pt) node[anchor=north west]{$o_{10}$};
		
		\filldraw[black, xscale=2/sqrt(3),xslant=.5] (0.7,1.15) circle (0pt) node[anchor=north]{$R_{*}$};

		\filldraw[black, xscale=2/sqrt(3),xslant=.5] (2.6,-0.55) circle (0pt) node[anchor=north]{$H_0$};
		\filldraw[black ,xscale=2/sqrt(3),xslant=.5] (2.4,1.2) circle (0pt) node[anchor=west]{$H_2$};
		\filldraw[black, xscale=2/sqrt(3),xslant=.5] (1.2,2.4) circle (0pt) node[anchor=west]{$H_1$};

		\filldraw[black, xscale=2/sqrt(3),xslant=.5] (0.6,0.5) circle (1pt) node[anchor=north west]{$p$};
		
	\end{tikzpicture}
\caption{The belief regions and the set $R_{*}$ for $n=2$.}
\label{inefficient}
\end{figure}

We show that, whenever there are at least two risky actions, the Bayesian persuasion outcome violates the support bound in \Cref{necessary} for a full-dimensional set of priors.\footnote{Throughout this subsection, genericity is understood relative to the class of payoff profiles satisfying the maintained assumptions. These payoff restrictions do not impose the indifference conditions underlying \Cref{necessary}, so the genericity argument continues to apply within this class.}

\begin{proposition}\label{BPfixedpref}
Suppose $n\geq 2$. Then there exists a full-dimensional set $R_*\subseteq\Delta(\Omega)$ such that, generically, the BP outcome $\mu_p^*$ is inefficient for every $p\in\mathrm{int}\,R_*$.
\end{proposition}

For illustration, suppose $n=2$ and the prior $p\in\operatorname{int}R_*$ is as shown in \Cref{inefficient}. Any signaling policy supported on the outer points must induce all three actions. Thus, an optimal policy must assign positive weight to $\omega_0$. However, it cannot assign positive weight to $\omega_i$ for $i=1,2$. If $\omega_i$ were used, one could shift weight from $\omega_i$ to $o_{i0}$, while decreasing the weight on $\omega_0$ to preserve Bayes plausibility. This increases the probability of risky action $a_i$ at the expense of the safe action $a_0$, thereby strictly increasing the sender's payoff. Consequently, either all three actions are played in state $\omega_0$, or randomization occurs in at least two states. Hence, the BP outcome is necessarily mixed and therefore violates the support bound in \Cref{necessary}. A similar argument extends to any number $n\geq 2$ of risky actions and yields a full-dimensional set of priors for which the BP outcome violates the support bound.

The previous proposition established inefficiency for a range of priors under fixed preferences. We now fix the interior prior and vary the receiver's preferences. Specifically, we consider threshold best responses.

Consider a receiver who takes risky action $a_i$ only when his posterior belief that the state is $\omega_i$ is at least $T>0.5$. A similar threshold-based persuasion environment is considered by \citet{aybas2026persuasion}. The corresponding best response regions are then given by
\begin{align*}
C_i &:= \{q\in\Delta(\Omega):q(\omega_i)\geq T\},
\qquad i=1,\ldots,n,\\
C_0 &:= \{q\in\Delta(\Omega):q(\omega_i)\leq T
\text{ for all } i=1,\ldots,n\}.
\end{align*}
Let $\mathcal{P}_T:=\{C_0,\ldots,C_n\}$ denote the posterior cover induced by threshold $T$.

We show that, with at least two risky actions, the BP outcome is mixed for sufficiently high thresholds.

\begin{proposition}\label{BPfixedprior}
Suppose $n\geq 2$. For every prior $p\in\operatorname{int}\Delta(\Omega)$, there exists $T_p<1$ such that, for every $T\in(T_p,1)$, the BP outcome is generically inefficient with respect to $\mathcal P_T$.
\end{proposition}

For illustration, suppose $n=2$ and a prior $p  \in  \mathrm{int}\, \Delta (\Omega)$. As $T$ increases, the prior eventually lies in a region where any signaling policy supported on the outer points must induce all three actions $a_0,a_1$, and $a_2$. By the previous  arguments, for any such prior, the Bayesian persuasion outcome must be mixed. The same logic extends to more than two risky actions.

Together, \Cref{BPfixedprior,BPfixedpref} show that   Bayesian persuasion outcomes are generically inefficient across a wide range of priors and preferences within the safe and risky actions environment.

\section{Conclusion}
\label{sec:conclusion}

This paper studied Pareto efficiency in games with incomplete information. Our main observation is that excessive randomization over actions leads to inefficiency. Efficiency requires that all state-contingent outcomes be supported by a common vector of welfare weights. When several actions are used across states, this requirement imposes multiple indifference conditions that generically cannot be satisfied simultaneously.

In finite state spaces, we derived a simple necessary condition for efficiency: generically, the total number of actions played across states must be strictly less than the sum of the number of players and the number of states. With a nonatomic prior, generically, every efficient outcome is pure almost surely. Across the three applications, the same tension emerges. Randomization can help satisfy incentive and equilibrium requirements, yet the resulting indifferences are difficult to reconcile with Pareto efficiency.

\subsection*{Declaration of AI-assisted technologies in the writing process}

During the preparation of this work, the authors used ChatGPT to improve the readability and clarity of the manuscript.

\bibliographystyle{apalike} 
\bibliography{references} 

\appendix
\section{Appendix}
\label{sec:appendix}

\subsection{Proof of \Cref{prop:exanteexpost}}

Suppose  that $\mu$ is efficient. Since $F_p$ is a finite convex polytope, the efficient payoff vector $u(\mu)$ is supported by some strictly positive vector $n\in\mathbb R^k_{++}$, that is, $n^\top u(\mu)=\max_{z\in F_p} n^\top z$. Intuitively, if no strictly positive supporting vector existed, one could move from $u(\mu)$ in a direction that
weakly improves all players and strictly improves at least one, contradicting Pareto efficiency.

We claim that the same vector $n$ supports $u(\mu\mid\omega)$ in every state. Suppose not, then for some state $\omega^*$, there exists $y\in F_{\omega^*}$ with $n^\top y>n^\top u(\mu\mid\omega^*)$. Replacing $u(\mu\mid\omega^*)$ by
$y$ and leaving all other state-contingent payoffs unchanged gives the feasible ex-ante payoff
\[
z=p(\omega^*)y+\sum_{\omega\neq\omega^*}p(\omega)u(\mu\mid\omega).
\]
Since $p(\omega^*)>0$, we have $n^\top z>n^\top u(\mu)$, contradicting the maximality of $u(\mu)$ over $F_p$. Hence,
$u(\mu\mid\omega)\in\argmax_{x\in F_\omega} n^\top x$ for every state $\omega$.

Conversely, suppose there exists $n\in\mathbb R^k_{++}$ such that $u(\mu\mid\omega)\in\argmax_{x\in F_\omega} n^\top x$ for every state $\omega$. Take any feasible payoff vector $z\in F_p$. Since $F_p=\sum_{\omega\in\Omega}p(\omega)F_\omega$, we can write $z=\sum_{\omega\in\Omega}p(\omega)x_\omega$ for some $x_\omega\in F_\omega$.
By state-wise optimality, $n^\top x_\omega\leq n^\top u(\mu\mid\omega)$ for every $\omega$. Summing over states gives $n^\top z\leq n^\top u(\mu)$. Thus, $u(\mu)$ maximizes $n^\top x$ over $F_p$. Since $n\in\mathbb R^k_{++}$,
this implies that $u(\mu)$ is Pareto efficient. Hence, $\mu$ is efficient.

\subsection{Efficiency via payoff deviations}
In the following proposition, we provide a necessary and sufficient condition for an outcome to be efficient.  To check for efficiency, one must look at the possible payoff changes from deviating to any action  in each state. Define
\begin{equation*}
	d_\mu(\omega,a):=\big( u_1(\omega,a)- u_1(\mu\mid\omega), \ldots, u_k(\omega,a)- u_k(\mu\mid\omega) \big)
\end{equation*}
as the deviation in state $\omega$ when action  $a$ is taken instead of $\mu(\omega) \in \Delta(A)$. Let
\begin{equation*}
	\mathcal{D}_\mu=\{ \tilde{d}\in \mathbb{R}^k:\tilde{d}=d_\mu(\omega,a) \text{ for some } \omega \in \Omega, a \in A \}
\end{equation*}
denote the set of all deviations given outcome $\mu$. Let
\begin{equation*}
	\mathrm{cone}(\mathcal{D}_\mu) := \Big\{ \sum_{d\in \mathcal{D}_\mu} \lambda_d \; d :\lambda_d \ge 0 \text{ for all } d \in \mathcal{D}_\mu \Big\}
\end{equation*}
denote the cone generated by the set of deviations $\mathcal{D}_\mu$. An outcome is efficient if and only if no nonnegative combination of deviations across states leads to a Pareto improvement.

\begin{proposition}
	\label{propdeviations}
	An outcome $\mu\colon\Omega \to \Delta(A)$ is efficient if and only if 
	\begin{equation} \label{cone}
		\mathrm{cone}(\mathcal{D}_\mu)\cap \mathbb{R}^k_{+}=\{0\}.
	\end{equation}
\end{proposition}

\begin{proof}
	$(\Rightarrow)$ We prove by contradiction. If $\mu$ is efficient, then by \Cref{prop:exanteexpost} there exists $n \in \mathbb R^k_{++}$ such that $n^\top d \leq 0$ for all $d \in \mathcal D_\mu$. Suppose instead that \eqref{cone} does not hold. Then there exist non-negative coefficients $(\lambda_d)_{d\in \mathcal D_\mu}$ with $d_\lambda = \sum_{d\in \mathcal D_\mu} \lambda_d d \in \mathbb R^k_{+}\setminus\{0\}$. Since $n \in \mathbb R^k_{++}$, this implies $n^\top d_\lambda > 0$. But $n^\top d_\lambda = \sum_{d\in \mathcal D_\mu} \lambda_d (n^\top d)\leq 0$, a contradiction.
	
	$(\Leftarrow)$ Assume \eqref{cone} holds, i.e., no nonzero nonnegative vector lies in the cone generated by the deviations. Let $B$ denote the $k \times |\mathcal{D}_\mu|$ matrix whose columns are the deviation vectors $d_\mu(\omega,a)$. By Mangasarian’s Theorem, as stated in \cite{perng2017class}, for any real matrix $B$ exactly one of the following holds:
	\begin{enumerate}
		\item[(i)] There exists $x\geq 0$, $x\neq 0$, such that $Bx \geq 0$, with at least one strictly positive component.
		\item[(ii)] There exists $n>0$, strictly positive componentwise, such that $n^\top B \leq 0$.
	\end{enumerate}   
	Since case (i) is ruled out by \eqref{cone}, case (ii) must hold. Thus, there exists $n \in \mathbb R^k_{++}$ such that $n^\top d \le 0$ for all $d \in \mathcal D_\mu$.
	
	Fix any state $\omega$ and action $a \in A$. Then $n^\top d_\mu(\omega,a)=n^\top (u(\omega,a) - u(\mu \mid \omega)) \leq 0$. Hence $n^\top u(\omega,a)\leq n^\top u(\mu\mid\omega)$ for every $a\in A$, so
	$u(\mu\mid\omega)\in \argmax_{x\in F_\omega} n^\top x$. The same strictly positive $n$ works for every state. Thus, by \Cref{prop:exanteexpost}, the outcome $\mu$ is efficient.
\end{proof}

\begin{remark}
	For two players, the condition reduces to the following: (i) there is no deviation with $d_\mu(\omega,a) \geq 0$, strict in at least one component, and (ii) for any pair of deviations $(d_1, d_2)$ where $d_1$ benefits player 1 and $d_2$ benefits player 2, the clockwise angle between $d_1$ and $d_2$ is at most $180^\circ$. Equivalently, $d_2^\top d_1^\dagger \geq 0$, where $d_1^{\dagger}$ denotes the vector obtained by rotating $d_1$ clockwise by $90^\circ$. In particular, if $d_1=(x,-y)$ with $x,y \geq 0$, then $d_1^{\dagger} = (-y,-x).$
\end{remark}

Although the deviations $d_\mu(\omega,a)$ are written using $u(\mu\mid\omega)$, efficiency depends only on the support of the outcome, not on the exact randomization probabilities or on the  interior prior. This insight is useful for Bayesian persuasion, where the support of the equilibrium outcome is fixed within convex regions of priors. Hence, within each such region, the efficiency status of BP outcomes remains unchanged.

\subsection{Proof of Proposition \ref{prop:exanteexpost-infinite}}

Since $F_p$ is compact and convex, an efficient payoff vector is supported by some nonzero nonnegative vector. Hence, if $\mu$ is efficient, there exists $n\in\mathbb R^k_+\setminus\{0\}$ such that
\[
 n^\top u(\mu)=\max_{z\in F_p} n^\top z.
\]
By the support function property of the Aumann integral \citep{aumann1965integrals,hess2002set},
\begin{equation}
\label{eq:support-aumann-main}
 \max_{z\in F_p} n^\top z
 =\int_\Omega \max_{x\in F_\omega}n^\top x\,dp(\omega).
\end{equation}
In our finite action setting, this identity  says that maximizing the linear functional $x\mapsto n^\top x$  over the ex-ante feasible set is equivalent to maximizing it state by state and then integrating. Moreover,
\[
 n^\top u(\mu)=\int_\Omega n^\top u(\mu\mid\omega)\,dp(\omega),
\]
and for every $\omega$,
\[
 n^\top u(\mu\mid\omega)\leq \max_{x\in F_\omega}n^\top x.
\]
Combining this inequality with \eqref{eq:support-aumann-main} gives
\[
 \int_\Omega\left[\max_{x\in F_\omega}n^\top x-n^\top u(\mu\mid\omega)\right]dp(\omega)=0.
\]
The integrand is nonnegative, so it is zero $p$-almost surely. Therefore $u(\mu\mid\omega)$ maximizes $n^\top x$ over $F_\omega$ for $p$-almost every state $\omega$.

\subsection{Prevalence argument for \Cref{thm:infinite}}
\label{app:no-ties-prevalence}

We use the notion of finite prevalence following
\citet{hunt1992prevalence} and \citet{anderson2001genericity}. A subset
$E$ of $\mathcal U$ is \emph{finitely shy} if there exists a
finite-dimensional subspace $V\subseteq\mathcal U$ such that, for every
$u\in\mathcal U$,
\[
\{v\in V:u+v\in E\}
\]
has Lebesgue measure zero in $V$. A set is \emph{finitely prevalent} if
its complement is finitely shy. Every finitely prevalent set is
prevalent. Finite unions of finitely shy sets are finitely shy.

For distinct actions $a,b\in A$, let
\[
E_{a,b}
:=
\left\{
u\in\mathcal U\colon
\text{there exists } n\in\mathbb R_+^k\setminus\{0\}
\text{ such that }
p\Bigl(
\bigl\{\omega\in\Omega:
n^\top u(\omega,a)=n^\top u(\omega,b)
\bigr\}
\Bigr)>0
\right\}.
\]

\begin{proposition}
\label{prop:infinite-prevalence}
Suppose $p$ is nonatomic. Then, for any distinct actions $a,b \in A$, the set $E_{a,b}$ is finitely shy.
\end{proposition}

\begin{proof}
Since $p$ is nonatomic, there exists a measurable map
$T:\Omega\to[0,1]^k$ whose distribution under $p$ is uniform on
$[0,1]^k$.

Choose distinct constants $(c_a)_{a\in A}$ and define
\[
v(\omega,a):=c_aT(\omega),
\qquad
V:=\{\lambda v:\lambda\in\mathbb R\}.
\]

Fix $u\in\mathcal U$ and define
\[
d_\lambda(\omega)
:=
u(\omega,a)-u(\omega,b)
+\lambda(c_a-c_b)T(\omega).
\]
We show that, for almost every $\lambda\in\mathbb R$ and every
$n\in\mathbb R_+^k\setminus\{0\}$,
\[
p\bigl(\{\omega:n^\top d_\lambda(\omega)=0\}\bigr)=0.
\]

Let $\omega_1,\ldots,\omega_k$ be independently distributed according
to $p$, and let
\[
D:=
[u(\omega_1,a)-u(\omega_1,b),\ldots,
 u(\omega_k,a)-u(\omega_k,b)]
\]
and
\[
X:=
[T(\omega_1),\ldots,T(\omega_k)].
\]
Since the columns of $X$ are independently and uniformly distributed on
$[0,1]^k$, $X$ is invertible almost surely.

For every realization for which $X$ is invertible,
\[
\lambda\longmapsto
\det\bigl(D+\lambda(c_a-c_b)X\bigr)
\]
is a nonzero polynomial, since its leading coefficient is
$(c_a-c_b)^k\det(X)\neq0$. Hence, it vanishes for only finitely many
values of $\lambda$. By Fubini's theorem, for almost every $\lambda$,
\[
\det
\begin{bmatrix}
d_\lambda(\omega_1)&\cdots&d_\lambda(\omega_k)
\end{bmatrix}
\neq0
\]
for $p^k$-almost every
$(\omega_1,\ldots,\omega_k)$.\footnote{A similar scalar-perturbation and
Fubini argument appears in \cite{jahedi2022robustness}.}

Fix such a $\lambda$. Suppose, toward a contradiction, that for some
$n\in\mathbb R_+^k\setminus\{0\}$,
\[
p\bigl(\{\omega:n^\top d_\lambda(\omega)=0\}\bigr)>0.
\]
Then, with positive probability, the $k$ independent vectors
$d_\lambda(\omega_1),\ldots,d_\lambda(\omega_k)$ all belong to the
hyperplane
\[
\{x\in\mathbb R^k:n^\top x=0\}.
\]
Hence, they are linearly dependent, contradicting the determinant
condition above.

Therefore, for almost every $\lambda$ and every
$n\in\mathbb R_+^k\setminus\{0\}$,
\[
p\bigl(\{\omega:n^\top d_\lambda(\omega)=0\}\bigr)=0.
\]
Equivalently, for every $u\in\mathcal U$,
\[
\{\lambda\in\mathbb R:u+\lambda v\in E_{a,b}\}
\]
has Lebesgue measure zero. Hence, $E_{a,b}$ is finitely shy.
\end{proof}

\subsection{Proof of Proposition \ref{prop:miralles}}

Propositions 3 and 4 of \citet{miralles2012cardinal} show that the auction mechanism for $k=2$ and the insured auction mechanism for $k\geq3$ reproduce the interim allocation probabilities of the optimal allocation rule. They therefore induce the same ex-ante payoff vector, so it suffices to establish inefficiency using these implementations.

Suppose, toward a contradiction, that the induced allocation rule is
efficient. By \Cref{prop:exanteexpost-infinite}, there exists a common vector
\[
n=(n_1,\ldots,n_k)\in\mathbb R_+^k\setminus\{0\}
\]
such that every allocation in the support maximizes the same weighted sum of
the agents' payoffs for $p$-almost every valuation profile.

First consider $k=2$. Consider valuation profiles at which both agents prefer
object $1$, agent $1$ wins the auction for object $2$, and her bid is strictly
between $0$ and $1/2$. Such profiles occur with positive probability. The
winner receives object $2$ and retains positive probability of object $1$,
while the losing agent also receives positive probability of object $1$.
Hence, the support contains both the allocation in which agent $1$ receives
both objects and the allocation in which agent $1$ receives object $2$ and
agent $2$ receives object $1$.

Efficiency requires these two allocations to generate the same weighted
welfare. Thus,
\[
n_1\omega_{11}=n_2\omega_{21}.
\]
Since $F$ has a density, for any fixed nonzero $n$ this equality cannot hold
on a set of valuation profiles of positive probability. This contradicts
efficiency.

Now suppose $k\geq3$. Fix an agent $i$ and consider valuation profiles at
which all agents prefer object $1$ and agent $i$ wins the auction for object
$2$. Such profiles occur with positive probability. Under the insured auction,
object $2$ is assigned to agent $i$, while every losing agent $j\neq i$
receives object $1$ with strictly positive probability. Hence, for every
$j\neq i$, the support contains the allocation in which object $2$ is assigned
to $i$ and object $1$ is assigned to $j$.

Efficiency therefore requires, for every $j,\ell\neq i$,
\[
n_j\omega_{j1}=n_\ell\omega_{\ell1}.
\]
Since $F$ has a density, these equalities can hold on a set of valuation
profiles of positive probability only if $n_j=0$ for every $j\neq i$.
Repeating the argument with a different auction winner implies that all
components of $n$ are zero, contradicting $n\neq0$.

\subsection{Proof of Proposition \ref{cheaptalk}}
We prove by contradiction. Suppose there exists an efficient cheap talk outcome $\mu$ that is not pure. By \Cref{necessary}, for two players, generically this implies the outcome must be quasi-pure: there is a single state $\omega^*$ in which exactly two actions $a_1$ and $a_2$ are played with positive probability.

There are two cases. First, suppose there exist two messages $m_1,m_2\in M$ such that $\sigma(m_1\mid\omega^*)\sigma(m_2\mid\omega^*)>0$
and $\tau(\cdot\mid m_1)\neq\tau(\cdot\mid m_2)$. Since only $a_1$ and
$a_2$ are played with positive probability in state $\omega^*$, both
receiver lotteries are supported on $\{a_1,a_2\}$. The sender's
equilibrium condition therefore implies
\begin{equation}
\label{eq:ct-sender-indiff}
	u_S(\omega^*,a_1)=u_S(\omega^*,a_2).
\end{equation}
For the outcome to be ex-post efficient in state $\omega^*$, the
receiver must also be indifferent:
\begin{equation}
\label{eq:ct-receiver-indiff}
	u_R(\omega^*,a_1)=u_R(\omega^*,a_2).
\end{equation}
Otherwise, shifting probability toward the receiver's preferred action
would make the receiver strictly better off without changing the
sender's payoff.

Second, suppose no two messages used in state $\omega^*$ induce
different lotteries. Since $\mu(\cdot\mid\omega^*)$ is not deterministic,
the receiver must mix between $a_1$ and $a_2$ after some on-path
message $m$. Because the outcome is pure in every state other than
$\omega^*$, this message can be sent only in state $\omega^*$.
Hence, the posterior after $m$ is degenerate at $\omega^*$, and the
receiver's equilibrium condition implies \eqref{eq:ct-receiver-indiff}. Ex-post efficiency then
requires \eqref{eq:ct-sender-indiff}; otherwise, shifting probability toward the sender's
preferred action would make the sender strictly better off without
changing the receiver's payoff.

Thus, efficiency requires both the sender's and the receiver's indifference
conditions to hold simultaneously, which is non-generic. Therefore,
generically, a cheap talk outcome is efficient only if it is pure.

\subsection{Proof of Proposition \ref{cheaptalk-transparent}}

$(\Rightarrow)$ First, suppose that the receiver mixes after some
on-path message. Let $q \in \Delta (\Omega)$ be the posterior induced by this message, and
let $a_1$ and $a_2$ be two actions played with positive probability,
where $u_S(a_1)>u_S(a_2)$. Since both actions are receiver best
responses at the belief $q$,
\[
\mathbb E_q[u_R(\omega,a_1)]
=
\mathbb E_q[u_R(\omega,a_2)].
\]
Hence, there exists a state $\omega^*\in\operatorname{supp}(q)$ such
that
\[
u_R(\omega^*,a_1)\geq u_R(\omega^*,a_2).
\]
Both actions are played with positive probability in state $\omega^*$.
Therefore, shifting probability from $a_2$ to $a_1$ strictly improves
the sender's payoff and weakly improves the receiver's payoff. Hence,
the outcome is not ex-post efficient in state $\omega^*$ and therefore
is not efficient.

Now suppose that the receiver takes a pure action after every on-path
message. Since the sender's payoff is state-independent and distinct
across actions, all on-path messages must induce the same action $a$.
Otherwise, the sender would strictly prefer a message that induces his
preferred action. Thus, the equilibrium outcome induces the same action
$a$ in every state.

If $a\neq a^*$, then, since $a^*$ is a receiver best response at some
belief $q$, there exists a state
$\omega^*\in\operatorname{supp}(q)$ such that
\[
u_R(\omega^*,a^*)\geq u_R(\omega^*,a).
\]
Replacing $a$ by $a^*$ in state $\omega^*$ strictly improves the sender's
payoff and weakly improves the receiver's payoff. Hence, the outcome is
not ex-post efficient and therefore is not efficient. Thus, any
efficient cheap talk outcome must induce $a^*$ with certainty.

$(\Leftarrow)$ If $a^*$ is induced with certainty, then the sender
obtains his unique highest feasible payoff. Hence, no Pareto improvement is
possible, and the outcome is efficient.

\subsection{Proof of Proposition \ref{BPfixedpref}}

For each $i\neq 0$, let $H_i$ be the hyperplane defined by
$\{\omega_0\}\cup_{j\neq i,0}\{o_{ji}\}$ (see Figure
\ref{inefficient}). The hyperplane $H_i$ separates the convex set
$C_i$ from $C_j$ for every $j\neq i,0$. Let $R_i$ be the half-space
determined by $H_i$ that contains $C_i$. For any prior $p\in \operatorname{int} \, R_i$,
every feasible policy  supported on $\operatorname{Out}(\mathcal P)$ must induce the risky action $a_i$.

Similarly, let $H_0$ be the hyperplane defined by
$\{o_{10},\ldots,o_{n0}\}$, and let $R_0$ be the half-space determined
by $H_0$ that contains $\omega_0$. For any prior $p\in \operatorname{int} \, R_0$, every
feasible policy supported on $\operatorname{Out}(\mathcal P)$ must induce the safe action $a_0$.

For any distinct $i,j\neq0$, let $R_{ij}:=R_i\cap R_j\cap R_0$.
The set $R_{ij}$ is nonempty because $\omega_0\in R_{ij}$. We next
show that $\mathrm{dim}\,R_{ij}=\mathrm{dim}\, \Delta(\Omega)=n$.



The points $o_{ij}$ and $o_{ji}$ are distinct because $C_i\cap C_j=\emptyset$. Let $z$ lie in the open segment joining $o_{ij}$ and $o_{ji}$. Then $z$ lies strictly in both $R_i$ and $R_j$. Moreover, $H_i$ and $H_j$ contain $\omega_0$, while $\omega_0$ satisfies the inequality defining $R_0$ strictly because $C_0$ contains an open neighborhood of $\omega_0$. Therefore, for sufficiently small $\varepsilon>0$, $$x_\varepsilon=(1-\varepsilon)\omega_0+\varepsilon z$$ lies strictly in $R_0$, $R_i$, and $R_j$. A sufficiently small perturbation of $x_\varepsilon$ preserves these strict inequalities. Hence, $R_{ij}$ is full-dimensional.

Let $\mathcal O:=\operatorname{Out}(P)$ and let
\[
\mathcal S:=
\bigcup_{\substack{B\subseteq\mathcal O\\ |B|\leq n}}
\operatorname{co}(B).
\]
Since $\mathcal O$ is finite and each set in this union has dimension
at most $n-1$, $\mathcal S$ has empty interior. Hence,
$\widehat R_{ij}:=\operatorname{int} \,R_{ij}\setminus\mathcal S$ is
nonempty and full-dimensional.

Fix $p\in\widehat R_{ij}$. Fix an optimal policy supported on the outer points.
Since $p\notin\mathcal S$, any such policy  requires at least $n+1$ distinct posteriors. 

Since $p\in \widehat  R_{ij}$, every signaling policy supported on $\operatorname{Out}(\mathcal P)$ induces $a_0$, $a_i$,
and $a_j$. We claim that an optimal policy cannot use the vertices
$\omega_i$ or $\omega_j$. Suppose, for example, that $\omega_i$ is
used. Since $a_0$ is induced, $\omega_0$ is also used. For sufficiently
small $\varepsilon>0$, reduce the weight on $\omega_i$ by
$\varepsilon$, increase the weight on $o_{i0}$ by
$\varepsilon/o_{i0}(\omega_i)$, and reduce the weight on $\omega_0$
by $\varepsilon o_{i0}(\omega_0)/o_{i0}(\omega_i)$. This preserves
Bayes plausibility and replaces some probability of $a_0$ with
$a_i$. Since $u_S(a_i)>u_S(a_0)$, it strictly increases the sender's
payoff, contradicting optimality. The same argument rules out
$\omega_j$.

These distinct posteriors account for at least $n+1$ actions used across states. Since $a_i$ is induced without using $\omega_i$, it must be induced through some posterior $o_{ir}$, where $r\neq i$. Hence, $a_i$ is played in both states $\omega_i$ and $\omega_r$, adding one more action to the count. The same argument for $a_j$ adds another. Therefore,
\[
\sum_{\omega\in\Omega}|\mu_p^*(\omega)|
\geq n+1+2=n+3=2+|\Omega|.
\]

To justify genericity within the maintained class, fix a collection of supports across states that violates the bound in \Cref{necessary}. As in the proof of \Cref{necessary}, let \(D_\mu\) collect the associated payoff-difference vectors. Since there are two players, \(D_\mu\) has at least two rows. The maintained assumptions restrict the payoff profile but do not force these payoff-difference vectors to be linearly dependent. Hence, except on a measure-zero subset of the remaining free payoff parameters, at least two of these vectors are linearly independent. Therefore, \(D_\mu\) has rank two, so no strictly positive common supporting vector exists. 

Thus, the BP outcome violates the support bound in
\Cref{necessary} and is therefore generically inefficient. Finally, define
\[
R_*
:=
\bigcup_{\substack{i,j\neq0\\ i\neq j}}\widehat R_{ij}.
\]
The set $R_*$ is full-dimensional, and the conclusion holds for every
$p\in R_*$.

 \subsection{Proof of Proposition \ref{BPfixedprior}}

Let $q$ denote the intersection of the line joining $\omega_0$ and $p$
with the face
\[
F_0:=\{r\in\Delta(\Omega):r(\omega_0)=0\}.
\]
For every risky action $a_i$, the prior belongs to $R_i$ whenever
\[
T\geq T_p^i:=1-q(\omega_i),
\]
and it belongs to $R_0$ whenever
\[
T\geq T_p^0:=1-p(\omega_0).
\]
Define
\[
T_p:=
\max\left\{
1-p(\omega_0),
1-\min_{i=1,\ldots,n}q(\omega_i)
\right\}.
\]
Because $p\in\operatorname{int} \, \Delta(\Omega)$, we have
$p(\omega_0)>0$ and $q(\omega_i)>0$ for every $i=1,\ldots,n$.
Hence $T_p<1$.  This is illustrated in \Cref{fig:partition}.

\begin{figure}[ht!] 
	\centering
	\hspace{3cm} 
	\begin{tikzpicture}
		\draw[fill=white, xscale=2/sqrt(3),xslant=.5]
		(3,0)-|(0,3)--cycle;
		\draw[fill=blue, fill opacity=0.50, xscale=2/sqrt(3),xslant=.5] 
		(3,0)-- (2,0)--(2,1)--(3,0);
		\draw[fill=red, fill opacity=0.50, xscale=2/sqrt(3),xslant=.5]  
		(0,3)-- (0,2) -- (1,2) -- (0,3);
		\filldraw[black, xscale=2/sqrt(3),xslant=.5] (2.2,0.65) circle (0pt) node[anchor=north]{$C_1$};
		\filldraw[black, xscale=2/sqrt(3),xslant=.5] (0.25,2.6) circle (0pt) node[anchor=north]{$C_2$};
		\filldraw[black, xscale=2/sqrt(3),xslant=.5] (0.6,1.2) circle (1pt) node[anchor=north]{$p$};

		\filldraw[black, xscale=2/sqrt(3),xslant=.5] (1,2) circle (1pt) node[anchor=west]{$q$};
		\filldraw[black, xscale=2/sqrt(3),xslant=.5] (0,0) circle (0pt) node[anchor=north]{$\omega_0$};
		\filldraw[black ,xscale=2/sqrt(3),xslant=.5] (3,0) circle (0pt) node[anchor=north]{$\omega_1$};
		\filldraw[black, xscale=2/sqrt(3),xslant=.5] (0,3) circle (0pt) node[anchor=south]{$\omega_2$};


		
		\draw[<->,xscale=2/sqrt(3),xslant=.5] (3.25,0.25) -- (1.25,2.25) node[anchor=south west ,pos=0.5]{$T_p^1=1-q(\omega_1)$};

		\draw[dotted,xscale=2/sqrt(3),xslant=.5] (0,0) -- (1,2);
		
		


	\end{tikzpicture}
	\caption{Threshold beliefs for $n=2$.}
	\label{fig:partition}
\end{figure}

For every $T>T_p$,
\[
p\in\operatorname{int}
\left(
R_0\cap\bigcap_{i=1}^nR_i
\right).
\]
Hence, every policy supported on $\operatorname{Out}(\mathcal P_T)$ induces the safe action $a_0$ and every risky action $a_i$.

Consider an optimal policy supported on the outer points. Since $a_0$
is induced, the posterior $\omega_0$ has positive weight. Moreover, no
posterior $\omega_i$, $i\neq0$, can have positive weight in an optimal
policy. Indeed, replacing $\omega_i$ by $o_{i0}$ and reducing the weight
on $\omega_0$ preserves Bayes plausibility and increases the probability
of the risky action $a_i$ at the expense of the safe action $a_0$.

 It follows that every risky action $a_i$ is played in at least two
states. Therefore, as $n \geq 2$,
\[
\sum_{\omega\in\Omega}|\mu_p^*(\omega)|
\geq 1+2n \geq 2+|\Omega|.
\]

The relative genericity argument given in the proof of Proposition~\ref{BPfixedpref} applies
here as well. The result now follows from Theorem \ref{necessary}.

\end{document}